\documentclass[graybox]{svmult}
\usepackage{mathptmx}       % selects Times Roman as basic font
\usepackage{helvet}         % selects Helvetica as sans-serif font
\usepackage{courier}        % selects Courier as typewriter font
\usepackage{type1cm}        % activate if the above 3 fonts are not available on your system
\usepackage{makeidx}         % allows index generation
\usepackage{graphicx}        % standard LaTeX graphics tool when including figure files
\usepackage{multicol}        % used for the two-column index
\usepackage[bottom]{footmisc}        % places footnotes at page bottom

\usepackage{amsfonts,amssymb,amsmath,amscd}

\def\R{\mathbb{R}}

\def\Hom{\text{\rm Hom}}

\def\IDE{\text{\rm id}} %%% the identity linear map (here = \id)
\def\di{\partial} %%% partial differential

\newcommand{\SVEC}[1]{\mathop{\raisebox{0.05pt}{$#1$}}\limits^{\raisebox{-4pt}{\tiny $\to$}}{}}

\def\Vx{\SVEC{\mathrm{x}}} 
\def\Vy{\SVEC{\mathrm{y}}}
\def\Vz{\SVEC{\mathrm{z}}}
\def\Vf{\SVEC{\mathrm{f}}}
\def\Vg{\SVEC{\mathrm{g}}}
\def\Vh{\SVEC{\mathrm{h}}}

\def\Ef{f} %%% used for maps V^n ->V
\def\Eg{g}
\def\Eh{h}

\def\Sef{\underline{f}} %%% sequences of \Ef
\def\Seg{\underline{g}}
\def\Seh{\underline{h}}

\def\Sy{{\mathbb S}} %%% permutation groups
\def\sm{$\mathbb{S}$-module} %%%
\def\PRT{\mathfrak{P}} %%% partitions
\def\SPRT{\text{\rm Part}}
\def\PP{\mathcal{P}} %%% a general operad
\def\cc{\gamma}      %%% operadic composition

\def\id{\text{\rm id}} %%% the operadic unit
\def\OEND{\text{\rm End}} %%% Endomorphisms operad

\def\FDIFF{\text{\rm FDiff}} %%% the group of formal diffeomorphisms

\def\BULL{\hspace{3pt}\bullet\hspace{3pt}} %%% the group operation for a symmetric operad
\def\STAR{\hspace{3pt}*\hspace{3pt}}       %%% the pre--Lie operation for a symmetric operad

\def\GRP{\mathfrak{G}}  %%% the symmetric part of the group related to a symmetric operad
\def\HGRP{\widehat{\GRP}} %%% the full group related to a symmetric operad
\def\LIE{\mathfrak{g}}  %%% the symmetric part of the Lie algebra related to a symmetric operad
\def\HLIE{\widehat{\LIE}}  %%% the full Lie algebra related to a symmetric operad

\def\Gf{\alpha} %%% coefficients 1 of a group element associated to an operad
\def\Gg{\beta} %%% coefficients 1 of a group element associated to an operad
\def\UGf{\underline{\Gf}} %%% group element 1 associated to an operad
\def\UGg{\underline{\Gg}} %%% group element 2 associated to an operad

\def\Lf{\mu} %%% coefficients 1 of a Lie algebra element associated to an operad
\def\Lg{\nu} %%% coefficients 1 of a Lie algebra element associated to an operad
\def\ULf{\underline{\Lf}} %%% Lie algebra element 1 associated to an operad
\def\ULg{\underline{\Lg}} %%%  Lie algebra element 2 associated to an operad

\newcommand{\beq}{\begin{equation}}
\newcommand{\eeq}{\end{equation}}
\newcommand{\beqa}{\begin{eqnarray}}
\newcommand{\eeqa}{\end{eqnarray}}
\newcommand{\beqs}{\begin{eqnarray*}}
\newcommand{\eeqs}{\end{eqnarray*}}
\newcommand{\nn}{\nonumber \\}
\def \podr {&& \hspace{0pt}}

\newcommand{\bnn}{\\ \nonumber}
\newcommand{\nnb}{\nonumber \\}

\newcounter{tmpc}
\newlength{\tmplenght}
\newlength{\tmplenghta}
\newlength{\tmplenghtb}
\newlength{\tmplenghtc}

\newenvironment{LIST}[1]{%
\setlength{\tmplenghta}{#1}
\setlength{\tmplenghtb}{#1}
\setlength{\tmplenghtc}{#1}
\advance\tmplenghtb-5pt
\advance\tmplenghtc 42pt
\setcounter{tmpc}{0}
\begin{list}{{\rm (\alph{tmpc})}}{\usecounter{tmpc}
\setlength{\leftmargin}{\tmplenghta}
\setlength{\rightmargin}{0cm}
\setlength{\itemsep}{1pt}
\setlength{\topsep}{3pt}
\setlength{\labelsep}{5pt}
\setlength{\labelwidth}{\tmplenghtb}
\setlength{\listparindent}{\tmplenghta}}
}{\end{list}}

\DeclareMathAlphabet{\mathbbm}{U}{bbm}{m}{n}

\DeclareSymbolFont{ltrs}     {OT1}{pzc}{m}{it}
\DeclareSymbolFont{ltrsa}     {OMS}{cmsy}{m}{n}
\DeclareSymbolFont{ltrsA}{U}{txmia}{m}{it}
\DeclareSymbolFont{symbolsC}{U}{txsyc}{m}{n}
\DeclareSymbolFont{ltrsB}{U}{rsfs}{m}{n}

\DeclareSymbolFontAlphabet{\mfrak}{ltrsA}
\DeclareMathAlphabet{\mathpzc}{OT1}{pzc}{m}{it}
\DeclareMathAlphabet{\mathrsfs}{U}{rsfs}{m}{n}

\def\GRPH{\Gamma} %%% graph
\def\VRT{\mathrm{vert}} %%% set of vertices
\def\FLG{\mathrm{flag}} %%% set of flags
\def\flg{f} %%% flag
\def\MPO{s} %%% map one
\def\MPS{\sigma} %%% second map
\def\VCL{\mathrm{Colv}} %%% a set of colors for the vertices
\def\FCL{\mathrm{Colf}} %%% a set of colors for the flags
\def\vcl{\mathrm{c}_{\text{\rm v}}} %%% map assigning colors for the vertices
\def\fcl{\mathrm{c}_{\text{\rm f}}} %%% map assigning colors for the flags
\def\enu{\nu} %%% the enumeration map
\def\GRF{\mathbb{K}} %%% the ground field (ring)
\def\Dgm{\text{\rm Dgm}} %%% the set of all diagrams
\def\UROP{\mathfrak{R}} %%% the universal ``contraction operad''
\def\JSE{J} %%% the set J
\def\PERM{\Sy} %%% the permutation group (= "\Sy")
\def\MND{\mathfrak{M}} %%% the monoid
\def\Mnd{\mathcal{M}} %%% the linear envelope of the monoid
\newcommand{\CVR}[4]{C_{{#1},{#3}}({#2},{#4})} %%% the symbol of a propagator/edge
\def\MON{M} %%% the monomial corresponding to a diagram

\def\GRF{\mathbb{K}} %%% the ground field (ring)
\def\Dgm{\text{\rm Dgm}} %%% the set of all diagrams
\def\UROP{\mathfrak{R}} %%% the universal ``contraction operad''
\def\JSE{J} %%% the set J
\def\PERM{\Sy} %%% the permutation group (= "\Sy")

\def\Adm{\mathrsfs{E}} %%% the set of admissible connections
\def\SYS{\Phi} %%% the defining rules for a physical theory
\def\TVER{\mathrsfs{V}} %%% set of type of vertices
\def\COPCON{\kappa} %%% coupling constant
\def\VCOPCON{\SVEC{\COPCON}}  %%% the vector of the coupling constants
\def\CHNG{\mathrm{K}} %%% formal change of constants
\def\VCHNG{\SVEC{\CHNG}}
\def\REGPAR{\varepsilon} %%% regularization parameter
\def\PHQ{U} %%% a physical ``quantity''
\def\RPHQ{\PHQ^{\text{\rm ren}}} %%% renormalized physical ``quantity''
\def\FINREN{\mathrm{X}} %%% finite change in the renormalization
\def\VFINREN{\SVEC{\FINREN}} %%% vector
\def\RMOR{\Xi} %%% the renormalization group action on operadic level
\def\RRMOR{\GRP(\Xi)} %%% the renormalization group action
\newcommand{\Wck}[2]{\text{\rm Wick}^{#1}_{#2}} %%% the Wick map
\def\ver{I} %%% a vertex

\def\Z{\mathbb{Z}} %%% the ring of integers

\begin{document}

\title*{Operadic construction of the renormalization group}
\titlerunning{Operadic construction of the renormalization group}
\author{Jean-Louis Loday  and Nikolay M. Nikolov}
\authorrunning{J.-L. Loday and N.M. Nikolov}
\institute{Jean-Louis Loday
\at 
Institut de Recherche Math\'ematique Avanc\'ee
CNRS et Universit\'e de Strasbourg,
Zinbiel Institute of Mathematics
%, 
%\email{name@email.address}
\and 
Nikolay M. Nikolov 
\at 
INRNE, Bulgarian Academy of Sciences, 
Tsarigradsko chaussee 72 Blvd., Sofia 1784, Bulgaria
%,
%\email{name@email.address}%
}
\maketitle

\abstract{First, we give a functorial construction of a group associated to a symmetric operad. Applied to the endomorphism operad it gives the group of formal diffeomorphisms. Second, we associate a symmetric operad to any family of decorated graphs stable by contraction. In the case of Quantum Field Theory models it gives the renormalization group. As an example we get an operadic interpretation of the group of ``diffeographisms'' attached to the Connes-Kreimer Hopf algebra.}

\section{Introduction}
\label{sec:1}

The  combinatorics underlying the renormalization of Quantum Field Theory (QFT) is encoded into the  Feynman diagrams.  The diagram technique is a powerful tool in perturbative QFT. It was discovered by Connes and Kreimer that the combinatorics in renormalization can be described by a Hopf algebra structure on the space of Feynman diagrams since the attached group is the renormalization group. In this paper our aim is to systematize this procedure by means of symmetric operads. First we show that a family of decorated graphs which is stable for the contraction of the internal edges determines a symmetric operad. Second, we show that to any symmetric operad is attached a (formal) group which takes care of the symmetric group action. Combining the two constructions we get the construction of a group attached to families of diagrams. In the case of QFT we get the renormalization group.

\medskip

For the notation and terminology on operads we follow \cite{LV11} for which we refer for details.

\section{Operadic construction of the group of formal diffeomorphisms}
\label{sec:2}

Let $V$ $\equiv$ $\R^N$ be a vector space and
$\Vx = (x_1,\dots,x_N), \Vy, \Vz$ $\in$ $V$.
Consider the formal power series
\beqa\label{ee1.0}
\Vy \, = \, \Vf \bigl(\Vx\bigr)
\, = \podr \mathop{\sum}\limits_{n \, = \, 1}^{\infty}
\frac{1}{n!} \,
\mathop{\sum}\limits_{\mu_1,\dots,\mu_n \, = \, 1}^{N}
\Vf_{\mu_1,\dots,\mu_n} \,
x_{\mu_1} \cdots x_{\mu_n} \,,
\bnn
\Vz \, = \, \Vg \bigl(\Vy\bigr)
\, = \podr \mathop{\sum}\limits_{n \, = \, 1}^{\infty}
\frac{1}{n!} \,
\mathop{\sum}\limits_{\mu_1,\dots,\mu_n \, = \, 1}^{N}
\Vg_{\mu_1,\dots,\mu_n} \,
x_{\mu_1} \cdots x_{\mu_n} \,,
\eeqa
where $\Vf_{\mu_1,\dots,\mu_n}$ $=$
$(f_{\nu;\mu_1,\dots,\mu_n})_{\nu \, = \, 1}^N$
and $\Vg_{\mu_1,\dots,\mu_n}$ $=$
$(g_{\nu;\mu_1,\dots,\mu_n})_{\nu \, = \, 1}^N$
are the series coefficients.
Since these series do not have constant terms (i.e., terms with $n=0$)
it is well known that their composition
\beq\label{COMP0}
\Vz \, = \, \Vg\bigl(\Vf \bigl(\Vx\bigr)\bigr)
\, = \, \mathop{\sum}\limits_{n \, = \, 1}^{\infty}
\frac{1}{n!} \,
\mathop{\sum}\limits_{\mu_1,\dots,\mu_n \, = \, 1}^{N}
\Vh_{\mu_1,\dots,\mu_n} \,
x_{\mu_1} \cdots x_{\mu_n} \,,
\eeq
can be determined completely algebraically.
A~less popular fact is the formula for the coefficients
$\Vh_{\mu_1,\dots,\mu_n}$ $=$
$(h_{\nu;\mu_1,\dots,\mu_n})_{\nu \, = \, 1}^N$
of the composition series:%
\beq\label{ee1.1}
h_{\nu;\mu_1,\dots,\mu_n} \, =
\mathop{\sum}\limits_{\PRT \, \in \, \SPRT \{1,\dots,n\}} \
\mathop{\sum}\limits_{\rho_1,\dots,\rho_k \, = \, 1}^{N} \
g_{\nu;\rho_1,\dots,\rho_k} \,
f_{\rho_1;\mu_{i_{1,1}},\dots,\mu_{i_{1,j_1}}} \hspace{-2pt}\cdots
f_{\rho_k;\mu_{i_{k,1}},\dots,\mu_{i_{k,j_k}}} 
\,,
\eeq
which, in the case $N=1$, is known as the \emph{Fa\`a di Bruno formula}.
Here are the notations used in Eq.~(\ref{ee1.1}):
\begin{LIST}{33pt}
\item[$\bullet$]
the sum is over all partitions
\beq\label{ee1.2}
\PRT \, = \,
\Bigl\{
\bigl\{i_{1,1},\dots,i_{1,j_1}\bigr\},
\dots,
\bigl\{i_{k,1},\dots,i_{k,j_k}\bigr\}
\Bigr\} \,
\eeq
of the set $\{1,\dots,n\}$;
\item[$\bullet$]
in particular, $k$ is the cardinality $|\PRT|$ of the partition $\PRT$
and $j_1$, $\dots,$ $j_k$ are the cardinalities of its pieces;
\item[$\bullet$]
the partitions $\PRT$ are unordered,
but we shall introduce a ``canonical order'' such that
inside each group the elements are in increasing order
and the groups are ordered according to the order of their
minimal elements
\beq\label{ee1.3}
i_{\ell,1} < \cdots < i_{\ell,j_{\ell}}
\, , \quad
i_{1,1} < i_{2,2} < \cdots < i_{k,j_k} \,.
\eeq
\end{LIST}
Note that all the coefficients $\Vf_{\mu_1,\dots,\mu_n}$,
$\Vg_{\mu_1,\dots,\mu_n}$ and $\Vh_{\mu_1,\dots,\mu_n}$
are symmetric in their indices $\mu_1,\dots,\mu_n$
and hence, our convention in Eq. (\ref{ee1.1})
about the order on $\PRT$ is not essential.
However, we shall see that dropping the symmetry condition
on the coefficients still defines an associative product.

Let us try to simplify a little bit Eq. (\ref{ee1.1}) by absorbing
some summations:
the coefficients $\Vf_{\mu_1,\dots,\mu_n}$ define a multi-linear
map
\beq\label{ee1.6}
\Ef_n \, = \, \bigl(\Vf_{\mu_1,\dots,\mu_n}\bigr)
\, : \, V^{\otimes n} \to V \,
\eeq
and vice versa, every multi-linear map
$\Ef_n : V^{\times n} \to V$ defines a system of coefficients
$\Vf_{\mu_1,\dots,\mu_n}$ by its matrix elements.
Furthermore, the coefficients $\Vf_{\mu_1,\dots,\mu_n}$
are symmetric in $\mu_1,\dots,\mu_n$ iff the map
$\Ef_n$ is symmetric.
Similarly, we set
$$
\Eg_n \, = \, \bigl(\Vg_{\mu_1,\dots,\mu_n}\bigr)
\, : \, V^{\otimes n} \to V
\, , \quad
\Eh_n \, = \, \bigl(\Vh_{\mu_1,\dots,\mu_n}\bigr)
\, : \, V^{\otimes n} \to V \,
$$
($n=1,2,\dots$).
Then Eq. (\ref{ee1.1}) reads
\beq\label{ee1.4}
\Eh_n \, =
\mathop{\sum}\limits_{\PRT \, \in \, \SPRT \{1,\dots,n\}} \
\Eg_k
\circ
\bigl(
\Ef_{j_1} \hspace{-1pt} \otimes \cdots \otimes
\Ef_{j_k}
\bigr)
\circ
\sigma_{\PRT}
\,,
\eeq
where the numbers $k,j_1,\dots,j_k$
are defined by conventions (\ref{ee1.2}) and (\ref{ee1.3})
together with the permutation $\sigma_{\PRT} \in \Sy_n$, which is
$$
\sigma_{\PRT} \, := \,
\bigl(
i_{1,1},\dots,i_{1,j_1},
\ldots,
i_{k,1},\dots,i_{k,j_k}
\bigr) \,.
$$

Thus, the formal power series $\Vy = \Vf \bigl(\Vx\bigr)$ of formula (\ref{ee1.0})
is encoded by a sequence
$$
\Sef \, = \, \bigl(\Ef_1,\Ef_2,\dots,\Ef_n,\dots\bigr)
\, \in \,
\mathop{\prod}\limits_{n \, = \, 1}^{\infty}
\Hom \bigl(V^{\otimes n},V\bigr)^{\Sy_n}
$$
($\Hom \bigl(V^{\otimes n},V\bigr)^{\Sy_n}$ being the subspace of $\Sy_n$--invariant maps in $\Hom \bigl(V^{\otimes n},V\bigr)$).
The multiplication in
\(\mathop{\prod}\limits_{n \, = \, 1}^{\infty}
\Hom \bigl(V^{\otimes n},V\bigr)^{\Sy_n}\),
$$
\Seh \, = \,
\Seg \BULL \Sef \, := \,
\bigl(h_n\bigr)_{n \, = \, 1}^{\infty}
\,,
$$
that is defined by Eq. (\ref{ee1.4}) is associative.
It has a unit, the composition unit:
$$
\underline{1} \, = \,
(\IDE_V,0,\dots)
$$
Furthermore, if we assume that $\Ef_1 = \IDE_V$ (the identity map of $V$),
then $\Sef$ has a composition inverse
$\Sef^{-1}$ $=$ $\bigl(1,(\Sef^{-1})_2,\dots\bigr)$ since for $n>1$
we have
$$
0 \, = \, \bigl(\underline{1}\bigr)_n \, = \,
\bigl(\Sef^{-1} \BULL \Sef \bigr)_n \, = \,
(\Sef^{-1})_n + \Ef_n + \text{low order terms} \,,
$$
which inductively fixes $(\Sef^{-1})_n$.

The so described \emph{group of formal diffeomorphisms} is denoted by
\beq\label{ee1.7}
\FDIFF (V) \, \cong \,
\{\IDE_V\} \times
\mathop{\prod}\limits_{n \, = \, 2}^{\infty}
\Hom \bigl(V^{\otimes n},V\bigr)^{\Sy_n} \,.
\eeq
Note that the vector space $V$ can be even arbitrary linear vector space:
$N$ then will be the cardinality (possibly, infinite) of the linear basis of $V$ and the series
(\ref{ee1.0}) would be neither more nor less formal.
We note also that $f_{\nu;\mu_1,\dots,\mu_n}$ for fixed $\mu_1,\dots,\mu_n$ are
nonzero only for no more than a finite number of indices $\nu$ since they
are coordinates of the vector $\Vf_{\mu_1,\dots,\mu_n}$.
Hence, the correspondence $\Vf \bigl(\Vx\bigr)$ $\leftrightarrow$ $\Sef$ defined by (\ref{ee1.6}) remains valid and the composition
(\ref{COMP0}) is again well defined algebraically.

\section{Group associated to a symmetric operad}\label{Se-OGRP}

We now observe that the multiplication (\ref{ee1.4}) has a straightforward generalization in a symmetric operad  (see Eq. (\ref{ee1.4gen}) below).
Indeed, it uses two basic structures which are axiomatized in the operad theory.
These are the composition of multilinear maps and the right action of (or, composition with) permutations.

\begin{theorem}\label{Th1.1}
{\rm (\cite{LN12})}
There is a functor together with a subfunctor:
\beqs
\begin{array}{ccl}
\left\{\hspace{0pt}
\begin{array}{c}\text{\it Category of} \\
\text{\it Symmetric operads}\end{array}
\hspace{0pt}\right\}
& \hspace{5pt}\to\hspace{7pt} &
\left\{\hspace{0pt}
\begin{array}{c}\text{\it Category of} \\
\text{\it Groups}\end{array}
\hspace{0pt}\right\}
\raisebox{-18pt}{}
\\
\PP \, = \, \{\PP(n)\}_{n \, = \, 1}^{\infty}
& \hspace{5pt}\mapsto\hspace{7pt} &
\HGRP(\PP) \, = \, \{\id\} \times
\mathop{\prod}\limits_{n \, = \, 2}^{\infty}
\PP(n)
\\
& & \hspace{10pt}\bigcup
\raisebox{0.5pt}{\hspace{-2pt}\small $\|$}
\raisebox{-10pt}{}
\\
\PP \, = \, \{\PP(n)\}_{n \, = \, 1}^{\infty}
& \hspace{5pt}\mapsto\hspace{7pt} &
\GRP(\PP) \, = \, \{\id\} \times
\mathop{\prod}\limits_{n \, = \, 2}^{\infty}
\PP(n)^{\Sy_n}
\,,
\end{array}
\eeqs
where $\PP(n)^{\Sy_n}$ stands for the subspace
of\hspace{3pt} $\Sy_n$--invariant elements.
The multiplication law is given by
\beq\label{ee1.4gen}
(\UGg \BULL \UGf)_n \, =
\mathop{\sum}\limits_{\PRT \, \in \, \SPRT \{1,\dots,n\}} \
\cc \bigl(
\Gg_k ;\Gf_{j_1},\dots,\Gf_{j_k}
\bigr)^{\sigma_{\PRT}}
\,
\eeq
for $\UGf = (\Gf_n)_{n \, = \, 1}^{\infty}$ and
$\UGg = (\Gg_n)_{n \, = \, 1}^{\infty}$
and the notations of Eq. $(\ref{ee1.4})$.
On operadic morphisms $\vartheta : \PP \to \PP'$
$(=\{\vartheta_n : \PP (n) \to \PP' (n)\}_{n \, = \, 1}^{\infty})$
the functor gives
$$
\HGRP (\vartheta) \, := \,
\mathop{\prod}\limits_{n \, = \, 1}^{\infty}
\vartheta_n \,.
$$
In the case of $\OEND_V$ we have a natural isomorphism
\beq\label{ISOV}
\GRP \bigl(\OEND_V\bigr) \, \cong \,
\FDIFF (V) \,.
\eeq
\end{theorem}

The most nontrivial part of the above statement is the associativity of the operation $\BULL$~(\ref{ee1.4gen}).
It can be proven by straightforward inspection.
The existence of a unit and inverse elements follows exactly by the same arguments as for the group of formal diffeomorphisms.

\begin{remark}\label{rm1}
There is a natural group associated with a non--symmetric operad $\PP = \bigl\{\PP_n\bigr\}_{n\geqslant 1}$
(see \cite[Sect. 5.8.15]{LV11}).
However when this construction is applied to a symmetric operad considered as a non-symmetric it gives a different group.
\end{remark}

We will give below some facts about the structure of of the groups related to symmetric operads.

\begin{proposition}\label{Pr2}
{\rm (\cite{LN12})}
Let us set for $m > 0$ 
$$
\HGRP_{m}(\PP) \, = \, 
\Bigl\{
\UGf = (\Gf_n)_{n \, = \, 1}^{\infty} \in \HGRP(\PP)
\Bigl| \Gf_2 = \cdots = \Gf_{m} = 0
\Bigr\}
$$
$($for $m=1$, $\HGRP_{1}(\PP):=\HGRP(\PP)$$)$.
Then $\HGRP_{m}(\PP)$ is a normal subgroup of $\HGRP(\PP)$.
\end{proposition}

Note that
$$
\HGRP(\PP) \, = \, \mathop{\lim}\limits_{\longleftarrow} \, \HGRP(\PP) \bigl/ \HGRP_m(\PP)
$$
and in the case when the operadic spaces $\PP(n)$ are finite dimensional the quotient groups
are (finite dimensional) Lie groups.
Hence, in the latter case  the group $\HGRP(\PP)$ is a \emph{pro-Lie group}.
We use this fact to derive the Lie algebra corresponding to the group $\HGRP(\PP)$ together with the exponential map.

\begin{theorem}\label{th03}
{\rm (\cite{LN12})}
The Lie algebra corresponding to the group $\HGRP(\PP)$ is 
$$
\HLIE(\PP) \, = \, \{0\} \, \times \,
\mathop{\prod}\limits_{n \, = \, 2}^{\infty}
\PP(n)
$$
The Lie bracket on $\HLIE(\PP)$ is built from a pre-Lie bracket
$$
[\ULf,\ULg] \, = \, \ULf \STAR \ULg - \ULg \STAR \ULf
$$
$($$\ULf,\ULg \in \HLIE(\PP)$$)$, where\footnote{%
$\circ_i$ is the $i$th operadic partial composition}
\beqa\label{ee1.4gen2}
(\ULf \STAR \ULg)_n 
& = &
\mathop{\sum}\limits_{\emptyset \, \neq \, J \, \subseteq \, \{1,\dots,n\}} \
\bigr(
\Lg_k \, \circ_{\min \, J} \, \Lf_{j}
\bigr)^{\sigma_{\PRT_J}}
\nnb
& \equiv &
\mathop{\sum}\limits_{\emptyset \, \neq \, J \, \subseteq \, \{1,\dots,n\}} \
\cc \bigl(
\Lg_k ;\id,\dots,\id,\mathop{\Lf_{j}}\limits_{\mathop{}\limits^{\uparrow}_{\min \, J}},\id,\dots,\id
\bigr)^{\sigma_{\PRT_J}} 
\eeqa
where $j=|J|$ and the partition $\PRT_J$ is the partition $\bigl\{\{i\}\bigl| i \in \{1,\dots,n\}\backslash J\bigr\} \cup \{J\}$. $($Note that the sum in $(\ref{ee1.4gen2})$ is the subsum in $(\ref{ee1.4gen})$ corresponding to partitions $\PRT$ of a form $\PRT_J$.$)$
\end{theorem}

The Lie algebra  $\HLIE(\PP)$ is again an inverse limit of finite dimensional Lie algebras
$$
\HLIE(\PP) \, = \, \mathop{\lim}\limits_{\longleftarrow} \, \HLIE(\PP) \bigl/ \HLIE_m(\PP)
$$
where $\HLIE_m(\PP)$ is the ideal
$$
\HLIE_{m}(\PP) \, = \, 
\Bigl\{
\ULf = (\Lf_n)_{n \, = \, 1}^{\infty} \in \HLIE(\PP)
\Bigl| \Lf_2 = \cdots = \Lf_{m} = 0
\Bigr\} \,.
$$
Note that the quotient group $\HGRP(\PP) \bigl/ \HGRP_m(\PP)$ and Lie algebra
$\HLIE(\PP) \bigl/ \HLIE_m(\PP)$ are isomorphic as sets to the set
\(
\mathop{\prod}\limits_{n \, = \, 2}^{m} \PP(n)
\)
and the group and pre--Lie products on this set are just $\BULL$~(\ref{ee1.4gen}) and $*$~(\ref{ee1.4gen2}) truncated up to order $m$.

\section{Feynman diagrams and their combinatorics}
\label{Se3}

Feynman diagrams are a powerful tool in perturbation theory. 
They indicate the terms of perturbative expansions. 
Furthermore, many manipulation on the corresponding formal perturbation series have a combinatorial description by operations on diagrams. 

\bigskip

\noindent
{\bf a) Basic definitions} 

\medskip

\noindent
A Feynman diagram is a finite graph with various decorations.

A graph $\GRPH$ is a set of points, called \emph{vertices}, with attached \emph{flags}  (or \emph{half-edges}) to them. Some pairs of these flags are further joined to become edges connecting the corresponding vertices. All these structures are contained in the following data: two finite sets, the set of vertices $\VRT(\GRPH)$ and the set of flags
$\FLG(\GRPH)$, and two maps
\beq\label{strmap1}
\MPO: \FLG(\GRPH) \to \VRT(\GRPH)\, , \quad \MPS: \FLG(\GRPH) \to 
\FLG(\GRPH) 
\eeq
such that $\MPS^2 = \id$.
Thus, the map $\MPO$ represents the process of attaching flags to vertices, i.e., the flag $\flg$ is attached to the vertex $\MPO(\flg)$. The map $\MPS$ represents the process of joining flags, i.e., the flag $\flg$ is joined with the flag $\MPS(\flg)$. In the latter case if $\flg=\MPS(\flg)$ then we call this flag an external line; such a line is attached to one only vertex. If $\flg\neq\MPS(\flg)$ then the unordered pair $\{\flg,\MPS(\flg)\}$ form an edge, or an internal line of the graph, which is attached to the vertices $\MPO(\flg)$ and $\MPO(\MPS(\flg))$.  When $\MPO(\flg)=\MPO(\MPS(\flg))$ but $\flg\neq\MPS(\flg)$ we have an internal line attached to one and the same vertex. Such an internal line is called a tadpole and it is usually excluded to exist.

\medskip

To every graph we assign a topological space: its \emph{geometric realization}.
To this end we assign to each edge a copy of the closed interval $[0,1]$ (without the orientation) and to each vertex a point.
Then we glue all of these spaces according to the incidence between the edges and the vertices.

\medskip

A \emph{decorated graph} is a graph with some extra data. Forgetting these extra structure we obtain just a graph that is called the \emph{body} of the decorated graph. We shall consider graphs with the following decorations:

\medskip

a) \emph{Colors} for the vertices and for the flags. They form two sets
\begin{LIST}{33pt}
\item[--]
a set of colors for the vertices: $\VCL$
\item[--]
a set of colors for the flags: $\FCL$
\end{LIST}
Then we have maps assigning colors:
\beq\label{strmap2}
\vcl : \VRT(\GRPH) \to \VCL \,,\quad
\fcl : \FLG(\GRPH) \to \FCL \,.
\eeq

\medskip

b) The second type of decoration we shall consider is an \emph{enumeration} 
\beq\label{strmap3}
\enu: \VRT(\GRPH) \cong \{1,\dots,n\}
\eeq
of the set of vertices.

\bigskip

\noindent
{\bf b) Examples} 

\medskip

\noindent
These are the notion of graph and decorated graph, or also diagram. Here are some examples to illustrate them.

\begin{example}\label{exm01}
An example of a graph is: $\VRT(\GRPH)$ $=\{0,1\},$ $\FLG(\GRPH)$ $=$ $\{a,b,c,d,e,f\},$ $\MPO(a)$ $=$ $\MPO(b)$ $=$ $\MPO(c)$ $=$ $0,$ $\MPO(d)$ $=$ $\MPO(e)$ $=$ $\MPO(f)$ $=$ $1,$ $\MPS(a)$ $=$ $a,$ $\MPS(b)$ $=$ $e,$ $\MPS(c)$ $=$ $d,$ $\MPS(f)$ $=$ $f$. The geometric realization is:
\begin{center}
\raisebox{0pt}{\includegraphics[scale=0.39]{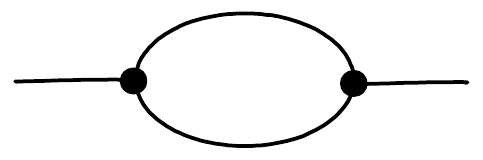}}
\end{center}
\end{example}

\begin{example}\label{exm02}
A decoration for the graph in Example \ref{exm01} is provided by $\VCL=\{\bullet\}$, $\FCL=$ $\bigl\{$\raisebox{-3pt}{\includegraphics[scale=0.25]{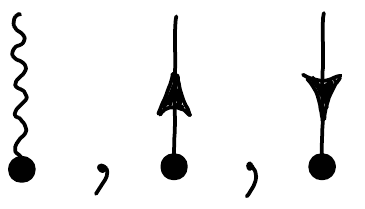}}$\bigr\}$, and coloring maps:
$\vcl(0)=\vcl(1)=\bullet$,  $\fcl(a)=$ \raisebox{-3pt}{\includegraphics[scale=0.25]{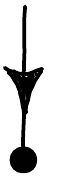}},
$\fcl(b)=$ \raisebox{-3pt}{\includegraphics[scale=0.25]{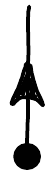}},
$\fcl(e)=$ \raisebox{-3pt}{\includegraphics[scale=0.25]{JLL-NMN_fig05c.pdf}},
$\fcl(f)=$ \raisebox{-3pt}{\includegraphics[scale=0.25]{JLL-NMN_fig05b.pdf}},
$\fcl(c)=$ \raisebox{-3pt}{\includegraphics[scale=0.25]{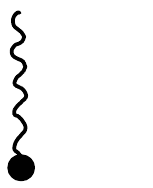}} $=\fcl(d)$.
The result can be drawn as
\begin{center}
\raisebox{0pt}{\includegraphics[scale=0.39]{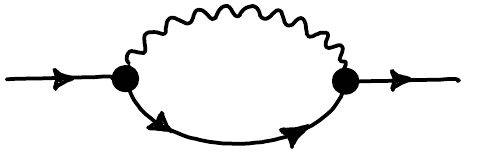}}
\end{center}
So, we indicated the colors in this example by shapes, which is common in physics.
Also if the colors of two joined flags coincide we indicate this as a color of the corresponding edge. 
In the above example we also meet situation of edges of the form \raisebox{0pt}{\includegraphics[scale=0.25]{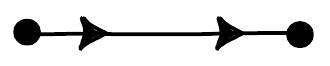}} and in this case it is also convenient to think of such an edge as an oriented edge \raisebox{0pt}{\includegraphics[scale=0.25]{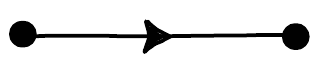}}. Then we can draw the diagram of this example as
\begin{center}
\raisebox{0pt}{\includegraphics[scale=0.39]{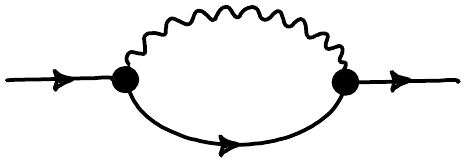}}
\end{center}
\end{example}

\bigskip

\noindent
{\bf c) Types of graphs and diagrams} 

\medskip

\noindent
A graph is called \emph{connected} if its geometric realization is a connected space.

Another important type of graphs are the so-called \emph{one particle irreducible} (1PI) graphs. A graph $\GRPH$ is called one particle irreducible if it is connected and after cutting any of its inner edges it remains connected.
Here cutting of an inner edge determined by a pair of flags $\flg\neq\MPS(\flg)$ means to change the second structure map $\MPS$ to a new map $\MPS'$ such that $\MPS'(\flg'):=\MPS(\flg)$ if $\flg' \neq \flg$ and $\flg'\neq\MPS(\flg)$, and $\MPS'(\flg'):=\flg'$ if $\flg'=\flg$ or $\flg'=\MPS(\flg)$.
We shall impose in addition the requirement that 1PI graphs have no tadpoles and have at least two vertices (or equivalently, at least one inner edge).

\medskip

If the body of a decorated graph is connected, then the graph is also called connected. Similarly a decorated graph is called 1PI if its body is 1PI.

\bigskip

\noindent
{\bf e) Operations on graphs and diagrams} 

\medskip

\noindent
A \emph{subgraph} of a graph $\GRPH$ is a subset $J \subseteq \VRT(\GRPH)$. It determines a graph $\GRPH_J$ as follows: the set of vertices of $\GRPH_J$ is $\VRT(\GRPH_J) := J \subseteq \VRT(\GRPH)$. The set of flags of $\GRPH_J$ is $\FLG(\GRPH_J) := \MPO^{-1} (J) \equiv \MPO^{-1} \bigl(\VRT(\GRPH_J)\bigr)$ and we set the map $\MPO_J : \FLG(\GRPH_J) \to \VRT(\GRPH_J)$ to be the restriction of the map $\MPO$. The map $\MPS_J : \FLG(\GRPH_J) \to \FLG(\GRPH_J)$ coincides with $\MPS$ whenever $\flg$ and $\MPS(\flg)$ belong to $\FLG(\GRPH_J)$: such pairs $\{\flg,\MPS(\flg)\}$ of different flags are the inner edges of the subgraph. For the remaining $\flg \in \FLG(\GRPH_J)$ we set $\MPS_J(\flg)=\flg$ and they are the outer edges of the subgraph.
Note that the outer edges of the graph $\GRPH_J$ are either outer edges of $\GRPH$ attached to a vertex in $J$ or they are inner edges of $\GRPH$ with only one end belonging to $J$.

If the graph $\GRPH$ is colored then the graph $\GRPH_J$  determined by a subgraph $J$ has an induced coloring defined just by the restrictions of the coloring maps $\vcl$ and $\fcl$ to $\VRT(\GRPH_J)$ and $\FLG(\GRPH_J)$, respectively.

If the graph $\GRPH$ is enumerated, then the graph $\GRPH_J$ has an induced enumeration provided by the unique monotonically increasing isomorphism $\enu (J) \cong \{1,\dots,|J|\}$.

\medskip

Another important operation on graphs is the \emph{contraction} of a subgraph.

For every graph $\GRPH$ and its subgraph $J \subseteq \VRT(\GRPH)$ we define the contracted graph $\GRPH/J$ as follows.
We introduce a new vertex $v_J$, which for the sake of definiteness can be identified with the set $J$. Then we set
\beqs
\VRT(\GRPH/J) \, := \podr
\bigl(\VRT(\GRPH) \backslash J\bigr) \cup \{v_J\} \,, \\
\FLG(\GRPH/J) \, := \podr
\bigl\{\flg \in \FLG(\GRPH) \, \bigl| \text{ if } \MPO(\flg) \text{ and } \MPO(\MPS(\flg)) \in J \text{ then } \flg=\MPS(\flg) \bigr\} 
\\ \, \equiv \podr
\FLG(\GRPH) \bigl\backslash \bigl\{\flg \in \FLG(\GRPH) \, \bigl| \, \MPO(\flg), \MPO(\MPS(\flg)) \in J \text{ and } \flg \neq \MPS(\flg)\bigr\}
\,,
\eeqs
in other words, $\FLG(\GRPH/J)$ contains all the flags of $\FLG(\GRPH)$ except those ones that form the inner edges of the graph $\GRPH_J$.
The structure maps $\MPO_{\GRPH/J}$ and  $\MPS_{\GRPH/J}$ are defined as follows:
\beqs
&
\MPO_{\GRPH/J}(\flg) \, := \, \MPO(\flg) \ \text{ if } \ \MPO(\flg) \notin J \quad \text{ and } \quad \MPO_{\GRPH/J}(\flg) \, := \, v_J \ \text{ if } \ \MPO(\flg) \in J \,,
& \\ &
\MPS_{\GRPH/J} \, := \, \MPS \bigl|_{\VRT(\GRPH/J)} \,, 
&
\eeqs
where the second identity is provided by the fact that $\FLG(\GRPH/J)$ is defined as a $\MPS$--invariant subset.
To summarize, the graph $\GRPH/J$ is obtained by shrinking all the vertices in $J$ to a single vertex $v_J$ and removing all the internal lines of $\GRPH_J$.
Note that if the graph $\GRPH$ is connected or 1PI, respectively, then so is $\GRPH/J$.

If the graph $\GRPH$ is colored, then for every pair $(J,L)$ consisting of a subset $J \subseteq \VRT(\GRPH)$ and an element $L \in \VCL$ we can define a colored contracted graph $\GRPH/(J,L)$ constructed as the graph $\GRPH/J$ endowed with the following coloring maps $\vcl'$ and $\fcl'$:
\beqs
&
\vcl'
\bigl|_{\VRT(\GRPH) \backslash J}
\, := \,
\vcl \bigl|_{\VRT(\GRPH) \backslash J}
\,,\qquad
\vcl'
(v_J) 
\, := \, L 
\,,
& \\ &
\fcl'
\, := \,
\fcl \bigl|_{\FLG(\GRPH/J)} \,.
&
\eeqs

Finally, if we have an enumerated graph $\GRPH$, then the contracted graph $\GRPH/J$ will be endowed with the enumeration provided by the unique monotonically increasing isomorphism
$$
\enu(\VRT(\GRPH) \backslash J) \cup \{\min \enu(J)\} \cong
\{1,\dots,n-|J|+1\} \,.
$$

\medskip

Note that if the graph $\GRPH$ has no tadpoles, then the graphs $\GRPH_J$ and $\GRPH/J$ have no tadpoles for every subgraph $J$ of $\GRPH$.

\bigskip

\noindent
{\bf f) Isomorphic diagrams} 

\medskip

\noindent
Let us introduce the notion of an \emph{isomorphism} of two enumerated diagrams $\GRPH$ and $\GRPH'$. We shall treat two such diagrams as identical. An isomorphism of graphs $\GRPH \cong \GRPH'$ consists of a pair of bijections $j_v:\VRT(\GRPH) \cong \VRT(\GRPH')$ and $j_f:\FLG(\GRPH) \cong \FLG(\GRPH')$, which commute with the structure maps $\MPO, \MPO'$ and $\MPS, \MPS'$, respectively. 
In other words, 
$j_v \circ \MPO = \MPO' \circ j_f$ and
$j_f \circ \MPS = \MPS' \circ j_v$.
An isomorphism of colored graphs is an isomorphism of graphs, which in addition satisfies
$\vcl = \vcl' \circ j_v$ and $\fcl = \fcl' \circ j_f$
(compatibility with the coloring maps).
Finally, an isomorphism of enumerated colored graphs is an  isomorphism of colored graphs which
preserves the enumeration.
Let 
\beqa\label{Dgm-def}
\Dgm(n) & := & \text{set of all equivalence classes of isomorphic }
\nnb & & \text{enumerated colored graphs with $n$ vertices}.
\eeqa

\bigskip

\noindent
{\bf g) Combinatorial Feynman rules, or, representation of diagrams in a monoid} 

\medskip

\noindent
There is a convenient one-to-one correspondence between the equivalence classes of isomorphic enumerated colored graphs and the elements (monomials) of a commutative monoid.
This construction follows on an abstract algebraic (or combinatorial) level the so called ``Feynman rules'' that assign in QFT to every Feynman diagram an analytic expression.
Let 
\beqa
\MND(n) & := & \text{the free commutative monoid with a set of generators}
\nn
\label{SG1}
& &
\bigl(\{1,\dots,n\} \times \VCL \bigr) \cup 
\bigl(\{1,\dots,n\} \times \FCL \bigr) \cup
\bigl(\{1,\dots,n\} \times \FCL \bigr)^{\times 2} .
\qquad
\eeqa
Let us introduce ``physical'' names and notation for the elements in the above three disjoint sets. We call the elements of $\FCL$ the basic ``fields'' and denote them by $\phi$, $\psi$, etc. Then the element $(i,\phi)$ $\in$ $\{1,\dots,n\} \times \FCL$ will be denoted by $\phi(i)$ and called a ``field at the point $i$''. 
Next, the elements $(i,\phi;j,\psi)$ $\in$ $\bigl(\{1,\dots,n\} \times \FCL \bigr)^{\times 2}$ will be denoted by $\CVR{\phi}{i}{\psi}{j}$ and will be called ``propagators''. 
Finally, the elements $L \in \VCL$ will be called ``interactions'' and a pair $(i,L)$ $\in$ $\{1,\dots,n\} \times \VCL$ will be called an interaction at the point $i$ and will be denoted by $L(i)$. 

Thus, in the above notations the set of generators (\ref{SG1}) for the monoid $\MND(n)$ reads:
\beqa\label{NSGEN}
&
\bigl\{L(i) \, \bigl| \, L \in \VCL, i =1,\dots,n\bigr\} \cup
\bigl\{\phi(i) \, \bigl| \, \phi \in \FCL, i=1,\dots,n\bigr\} 
& \nn &
\cup
\bigl\{\CVR{\phi}{i}{\psi}{j} \, \bigl| \, \phi,\psi \in \FCL, i,j = 1\dots,n\bigr\} \,.
&
\eeqa

Now, to each enumerated colored graph $\GRPH$ we assign a monomial in $\MND(n)$ in the following way.
To the vertex $\enu^{-1}(i)$ (i.e., to the vertex with number $i$) we assign $L(i)$ if its color is $L \in \VCL$. To each outer edge attached to the vertex $\enu^{-1}(i)$ we assign $\phi(i)$ if the color of the corresponding flag is $\phi \in \FCL$. To each inner edge connecting the vertices $\enu^{-1} (i)$ and $\enu^{-1}(j)$ we assign 
$\CVR{\phi}{i}{\psi}{j}$
if the colors of the flags attached to $\enu^{-1} (i)$ and $\enu^{-1}(j)$ are $\phi$ and $\psi$, respectively. Finally, we multiply all the above obtained generators in $\MND(n)$.
The resulting monomial in $\MND(n)$ is denoted by $\MON_{\GRPH}$.

\begin{example}\label{exm03}
In the case of Example \ref{exm02} with vertex enumeration $\enu(0)=1$, $\enu(1)=2$ we have
$$
\MON_{\GRPH} \, = \, \overline{\psi}(1) \, \psi(2) \, L(1) \, L(2) \, \CVR{A}{1}{A}{2} \, \CVR{\psi}{1}{\overline{\psi}}{2}
\,,
$$
where we denoted now the colors by letters:
$L:=\bullet\in\VCL$ and
$A$ $:=$ \raisebox{-3pt}{\includegraphics[scale=0.25]{JLL-NMN_fig05a.pdf}}\,,
$\psi$ $:=$ \raisebox{-3pt}{\includegraphics[scale=0.25]{JLL-NMN_fig05b.pdf}}\,,
$\overline{\psi}$ $:=$~\raisebox{-3pt}{\includegraphics[scale=0.25]{JLL-NMN_fig05c.pdf}}\,.
\end{example}

\begin{proposition}\label{Pro4.3}
The correspondence $\GRPH \mapsto \MON_{\GRPH}$ is a bijection $\Dgm(n) \cong \MND(n)$, i.e., it is a one-to-one correspondence between the equivalence classes of isomorphic enumerated colored graphs with $n$ vertices and the elements of the monoid $\MND(n)$.
\end{proposition}

\begin{proof}
It is clear that $\GRPH \mapsto \MON_{\GRPH}$ maps injectively the equivalence classes of diagrams to elements of $\MND(n)$. To see that this map is surjective one constructs for every element of $\MND(n)$ a diagram that 
reproduces this monomial.$\quad\Box$
\end{proof}

\section{The universal contraction operad}\label{S4.2n}

Recall that $\Dgm(n)$ is the set of all equivalence classes of isomorphic enumerated colored graphs with $n$ vertices.
Let us define
\beq\label{UROP}
\UROP (n) \, := \,
\text{Hom}_{\GRF} \Bigl(\GRF^{(\Dgm(n))}, \GRF^{(\VCL)}\Bigr)
\, \cong \,
\GRF^{\Dgm(n) \times \VCL} \,,
\eeq
where $\GRF^{(I)}$ stands for the vector space over the ground field (ring) $\GRF$ spanned by a basis indexed by $I$ and 
the existence of 
the second \emph{canonical} isomorphism 
follows in the case when $\VCL$ is a \emph{finite} set, which we shall assume further. This canonical isomorphism 
is provided by the decomposition
\beq\label{Frm}
Q (\GRPH) \, = \, 
\mathop{\sum}\limits_{L \, \in \, \VCL} 
q(\GRPH,L) \, L \,,
\eeq
where $Q \in \UROP(n)$. We shall treat the isomorphism at the second equality in (\ref{UROP}) as an identification,
$\UROP (n) = \GRF^{\Dgm(n) \times \VCL}$.

We call the elements of $\UROP(n)$ \emph{contraction maps}. This is motivated by the fact that they can be thought of as prescriptions for contracting subgraphs as we shall describe below.

Note that the action of the permutation group $\PERM_n$ on $\Dgm(n)$ induces an action on $\UROP(n)$.
We shall endow now the so-defined \sm\ $\UROP$ $=$ $\{\UROP(n)\}_{n \, \geqslant \, 1}$ with a structure of a symmetric operad. 

To this end we shall define the partial composition maps:
\beq\label{circ-i0}
\circ_i : \UROP(n) \otimes \UROP(j) \to \UROP(n-1+j) \,,
\eeq
$i=1,\dots,n$, $j=1,2,\dots$.
Let us introduce for every enumerated diagram $\GRPH$ the subsets of vertices $\JSE := \JSE(i,j) \subseteq \VRT(\GRPH)$: 
\beq\label{JVRT}
\JSE \ (\, \equiv J(i,j) ) := \,
\bigl\{\enu^{-1}(\ell) \, \bigl| \, \ell=i+1,\dots,i+j \bigr\} \,.
\eeq
We define for $Q'' \in \UROP(n)$, $Q' \in \UROP(j)$ and $\GRPH$ that is a representative of an isomorphism class in $\Dgm(n-1+j)$:
\beq\label{circ-i}
(Q'' \circ_i Q') (\GRPH) \, = \,
\mathop{\sum}\limits_{L \, \in \, \VCL} 
q'(\GRPH_J,L) \, Q''\bigl(\GRPH/(\JSE,L)\bigr) \,,
\eeq
where
\beq\label{QP}
Q'(\GRPH_{\JSE}) \, =: \, 
\mathop{\sum}\limits_{L \, \in \, \VCL} 
q'(\GRPH_{\JSE},L) \, L \,.
\eeq
Note that if we set
\beqa\label{QDP}
Q''(\GRPH'') \, = \podr
\mathop{\sum}\limits_{L \, \in \, \VCL} 
q''(\GRPH'',L) \, L 
\,, 
\quad 
\nn
Q (\GRPH) \, = \podr (Q''\circ_i Q') (\GRPH)
= \mathop{\sum}\limits_{L \, \in \, \VCL} 
q(\GRPH,L) \, L
,\quad
\eeqa
then Eq. (\ref{circ-i}) reads
\beq\label{circ-i1}
q(\GRPH,K) \, = \,
\mathop{\sum}\limits_{L \, \in \, \VCL} 
q'(\GRPH_J,L) \, q''\bigl(\GRPH/(\JSE,L),K\bigr) \,.
\eeq

\medskip

\begin{proposition}\label{Pro4.1}
{\rm (\cite{LN12})}
$\UROP$ $=$ $\{\UROP(n)\}_{n \, \geqslant \, 1}$ is a symmetric operad.
\end{proposition}

\medskip

The {\it proof} is straightforward checking and we omit it.

\section{Suboperads in $\UROP$ and concrete combinatorial models of Quantum Field Theory}\label{S4.3n}

In the previous section we have defined a universal operad $\UROP$ on decorated graphs, which can include, at the combinatorial level, any concrete model of Quantum Field Theory (QFT) provided that we have sufficiently many colors in $\VCL$ and $\FCL$.
So, the QFT models can be considered as particular suboperads of $\UROP$.
Describing these suboperads can be quite cumbersome in general and we shall do this in several steps.
At each step we shall impose certain restrictions on the contraction maps $Q \in \UROP(n)$.
These restrictions include, in particular, requirements that $Q$ should vanish on certain classes of diagrams that are ``not admissible for contraction''.

For instance, excluding tadpoles was a first example of such a restriction on diagrams. It was ``stable with respect to contractions and subdiagrams'' and hence, it defined a suboperad in $\UROP$.
More precisely, the statement is that the subspaces in $\UROP(n)$ for every $n=1,2,\dots$, which consist of those contraction maps that vanish on diagrams with tadpoles, form a suboperad.

Let us formulate the argument in a more general principle:

\medskip

\begin{proposition}\label{pro4x}
{\rm (\cite{LN12})}
Let $\SYS=\{\SYS(n)\}_{n \geqslant 1}$ be a system of subsets
$\SYS(n) \subseteq \Dgm(n)\times \VCL$ for $n=1,2,\dots$ and let us define
\beqa\label{RPHI}
\UROP_{\SYS}(n) \, = \podr
\, \GRF^{\SYS(n)} \, \subseteq \, \GRF^{\Dgm(n)\times \VCL} \, \equiv \, \UROP(n) \,,
\nn
\UROP_{\SYS}(n) \, \equiv \podr
\Bigl\{Q = \sum 
q \, L \in \UROP(n) \, \Bigl| \, q\bigl|_{(\Dgm(n) \times \VCL) \backslash \SYS(n)} \, = \, 0\Bigr\}
\,,
\eeqa
where we use the expansion (\ref{Frm}) 
and 
embeddings of type $\GRF^A \hookrightarrow \GRF^B$ for $A\subseteq B$, which are defined by $(x_a)_{a \,\in\, A} \mapsto (y_b)_{b \,\in\, B}$ such that $y_a = x_a$ for $a \in A$ and $y_b=0$ for $b \in B\backslash A$.

Then the following conditions are equivalent:
\begin{LIST}{33pt}
\item[$(i)$]
The system $\UROP_{\SYS} = \{\UROP_{\SYS}(n)\}_{n \geqslant 1}$ is a suboperad of $\UROP$.
\item[$(ii)$]
Each subset $\SYS(n)$ is $\PERM_n$-invariant and the system $\{\SYS(n)\}_{n \geqslant 1}$ has the property
\beqa\label{gnprin}
(\GRPH_J,L) \in \SYS(|J|)
\ \text{ and } \ 
(\GRPH/(J,L),K) \in \SYS(n-|J|+1)
\nn
\ \ \Rightarrow \ \  (\GRPH,K) \in \SYS(n) 
\eeqa
for every $\GRPH \in \Dgm(n)$, $J \subseteq \VRT(\GRPH)$ and $K,L \in \VCL$.
\end{LIST}
\end{proposition}

\begin{corollary}\label{Crl4x1}
The following systems form a suboperad in $\UROP$:
$$
\UROP_{\text{\rm 1PI}} (n) \, := \,
\bigl\{Q \in \UROP(n) \, \bigl| \, Q(\GRPH) = 0 \text{ if } \GRPH \text{ is \emph{not} 1PI 
} \bigr\}. 
$$
\end{corollary}

Let us give another example for a restriction on diagrams that induces a suboperad.
Let us consider a \emph{non-empty} subset
$$
\Adm \subset \FCL^{\times 2}
$$
and call it a set of \emph{admissible connections}. A colored graph $\GRPH$ is called $\Adm$--\emph{admissible} if for all flags $\flg \in \FLG(\GRPH)$ such that $\flg \neq \MPS(\flg)$ we have $(\fcl(\flg),\fcl(\MPS(\flg))) \in \Adm$. Or in other words, if the pairs of colors of the flags corresponding to the inner edges are contained in $\Adm$. 
As an application of Proposition \ref{pro4x} we get:

\begin{corollary}\label{Pro4.2}
Let $\Adm$ be any symmetric subset in $\FCL^{\times 2}$ and let $\UROP_{\Adm} (n)$ be the space that consists of all contraction maps $Q \in \UROP(n)$, which vanish on all diagrams that either are not $\Adm$--admissible, or have tadpoles. Then $\bigl\{\UROP_{\Adm} (n)\bigr\}_{n \, \geqslant \, 1}$ is a suboperad of $\UROP$.
\end{corollary}

Note that in Corollaries \ref{Crl4x1} and \ref{Pro4.2} the sets $\SYS(n)$ are of the form
$$
\SYS(n) = \Dgm'(n) \times \VCL
$$ 
for some subsets
$\Dgm'(n) \subseteq \Dgm(n)$.
In this case condition (\ref{gnprin}) reads
$$
\GRPH_J \in \Dgm'(|J|)
\ \text{ and } \ 
\GRPH/J \in \Dgm'(n-|J|+1)
\nn
\ \ \Rightarrow \ \  \GRPH \in \Dgm'(n)  \,.
$$
and $\UROP_{\SYS}$ is
\beqs
\UROP_{\SYS}(n) = \podr \bigl\{Q \in \UROP(n) \, \bigl| \, Q\bigl|_{\Dgm(n) \backslash \Dgm'(n)} = 0\bigr\}
\equiv \Hom_{\GRF} \Bigl(\GRF^{(\Dgm'(n))},\GRF^{\VCL}\Bigr) 
\\
= \podr \GRF^{\Dgm'(n) \times \VCL} \,.
\eeqs

\begin{example}\label{exm04}
Let us introduce an example of the set $\Adm$ for the case of Quantum Electrodynamics (QED).
In this case we use three colors for flags
$\FCL=$ $\bigl\{$\raisebox{-3pt}{\includegraphics[scale=0.25]{JLL-NMN_fig05.pdf}}$\bigr\}$
The set of admissible connections is:
$$
\Adm \, = \,
\Bigl\{ 
\bigl(\ \raisebox{-3pt}{\includegraphics[scale=0.25]{JLL-NMN_fig05a.pdf}}\ ,\ \raisebox{-3pt}{\includegraphics[scale=0.25]{JLL-NMN_fig05a.pdf}}\ \bigr),
\bigl(\ \raisebox{-3pt}{\includegraphics[scale=0.25]{JLL-NMN_fig05b.pdf}}\ ,\ \raisebox{-3pt}{\includegraphics[scale=0.25]{JLL-NMN_fig05c.pdf}}\ \bigr)
\bigl(\ \raisebox{-3pt}{\includegraphics[scale=0.25]{JLL-NMN_fig05c.pdf}}\ ,\ \raisebox{-3pt}{\includegraphics[scale=0.25]{JLL-NMN_fig05b.pdf}}\ \bigr)
\Bigr\} \,.
$$
The diagram of Example \ref{exm02} was thus $\Adm$--admissible for QED and as there we can use for edges single colors, one with no orientation and one with orientation. The non-oriented lines are called ``photon lines'' and the oriented lines are called ``electron lines''.
\end{example}

\medskip 

Our next ``selection rule'' for contraction maps is by the type of vertices.
A vertex is a colored graph with one vertex and no tadpoles. So, it contains only outer edges which are called \emph{corolla} of the vertex. 
The number of the external edges of the vertex is called its \emph{valency}.

Let $\TVER \subseteq \Dgm(1)$ be a set of vertices. We call the set $\TVER$ \emph{types of vertices} in the theory.
Let us define then the system $\SYS_{\TVER} = \{\SYS_{\TVER} (n)\}_{n \geqslant 1}$
\beqs
\SYS_{\TVER} (n) \, = \, 
\Bigl\{(\GRPH,L) \in \Dgm(n) \times \VCL \,\Bigl|\podr \forall J \subseteq \GRPH \, \bigl(\text{if } |J|=1 \text{ then } \GRPH_J \in \TVER\bigr) \text{ and }  
\\ \podr
\GRPH/(\VRT(\GRPH),L) \in \TVER\Bigr\} \,.
\eeqs
It follows that $\SYS_{\TVER}$ satisfies condition $(ii)$ of Proposition \ref{pro4x} and hence,
$$
\UROP_{\TVER} \, := \, \UROP_{\SYS_{\TVER}} \,,
$$
is a suboperad of $\UROP$.

Thus, a physical theory can be defined as intersection of the operads
\beq\label{RAV}
\UROP_{\Adm,\TVER} \, := \,
\UROP_{\text{\rm 1PI}} \cap \UROP_{\Adm} \cap \UROP_{\TVER} \,.
\eeq
In the next section we shall consider the main examples of physical theories.

\medskip

\begin{remark}\label{Rm4xx}
If $\{\SYS_{i}\}_{i \in I}$ is a collection of systems $\SYS_{i} = \{\SYS_{i} (n)\}_{n \geqslant 1}$ each satisfying condition $(ii)$ of Proposition \ref{pro4x} then 
$$
\mathop{\bigcap}\limits_{i \in I}
\UROP_{\SYS_{i}} \, = \, 
\UROP_{\SYS}
\quad \text{where} \quad
\SYS = \{\SYS(n)\}_{n \geqslant 1} \quad \text{with} \quad
\SYS(n) =
\mathop{\bigcap}\limits_{i \in I} \SYS_{i}(n) \,,
$$
and $\SYS$ also satisfies condition $(ii)$ of Proposition \ref{pro4x}.
\end{remark}
 
\section{The group related to the contraction operad and its representation in the group of formal diffeomorphisms on the space of interactions}\label{Se-X}

Having defined a symmetric operad $\UROP$ for each particular QFT model we have automatically a group $\GRP(\UROP)$ associated to it.
This group is precisely the operadic construction of the renormalization group.

\subsection{Notions of renormalization group}\label{Se-X1}

There are several widespread notions of renormalization group in physics and they do not lead to equal objects although they are closely related to each other.
We shall review below some of them.
For recent related works we refer the reader to \cite{S}, \cite{BDF}.

In renormalization theory a physical quantity $\PHQ$ (an observable for instance, or a correlation function in QFT) is derived as a function $\PHQ=\PHQ(\COPCON_1,\dots,\COPCON_N;\REGPAR)$ ($\equiv \PHQ(\VCOPCON;\REGPAR)$) of various parameters including: 
\begin{LIST}{33pt}
\item[$\bullet$]
physical constants $\COPCON_1,\dots,\COPCON_N$. In QFT these are called coupling constants.
\item[$\bullet$]
An additional subsidiary parameter $\REGPAR > 0$ called a regularization parameter. It makes meaningful the value of $\PHQ(\COPCON_1,\dots,\COPCON_N;\REGPAR)$ that is usually ill-defined for $\REGPAR\to 0$.
The latter limit corresponds exactly to the actual physical value of $\PHQ$ and the purpose of the renormalization is to understand how to do it.
\item[$\bullet$]
There might be further variables but we consider them as a ``part'' of $\PHQ$
(so that $\PHQ$ is then valued in some vector or function space).
\end{LIST}
Furthermore, in perturbation theory, one has defined $\PHQ$ only as a formal power series in the coupling constants
\beq\label{PHQexp}
\PHQ (\VCOPCON;\REGPAR) \, = \, 
\mathop{\sum}\limits_{n \, = \, 0}^{\infty}
\frac{1}{n!} \,
\mathop{\sum}\limits_{i_1,\dots,i_n \, = \, 1}^{N}
\PHQ_{i_1,\dots,i_n}(\REGPAR) \,
\COPCON_{i_1} \cdots \COPCON_{i_n} \,,
\eeq
with coefficients $\PHQ_{i_1,\dots,i_n}(\REGPAR)$ that are functions in $\REGPAR > 0$.
The renormalization issue now is to find such a change of the physical parameters:
\beq\label{CHNG}
\VCOPCON' \, = \, \VCHNG(\VCOPCON;\REGPAR)
\,,\qquad
\COPCON'_{i} \, = \, \mathop{\sum}\limits_{n \, = \, 1}^{\infty}
\frac{1}{n!} \,
\mathop{\sum}\limits_{i_1,\dots,i_n \, = \, 1}^{N}
\CHNG_{i;i_1,\dots,i_n}(\REGPAR) \,
\COPCON_{i_1} \cdots \COPCON_{i_n} \,,
\eeq
again as a formal power series,
so that after the substitution\footnote{%
in terms of formal power series; note that the series $\VCHNG(\VCOPCON;\REGPAR)$ starts from $n=1$ but for $\PHQ(\VCOPCON;\REGPAR)$ we do not have such a restriction}
\beq\label{PHQchng}
\RPHQ(\VCOPCON;\REGPAR) \, := \, \PHQ \bigl(\VCHNG(\VCOPCON;\REGPAR);\REGPAR\bigr)
\, = \,
\mathop{\sum}\limits_{n \, = \, 0}^{\infty}
\frac{1}{n!} \,
\mathop{\sum}\limits_{i_1,\dots,i_n \, = \, 1}^{N}
\RPHQ_{i_1,\dots,i_n}(\REGPAR) \,
\COPCON_{i_1} \cdots \COPCON_{i_n} \,,
\eeq
the resulting coefficients $\RPHQ_{i_1,\dots,i_n}(\REGPAR)$ would have a finite limit for $\REGPAR \to 0$.
We set the final renormalized physical quantity $\RPHQ$ to be
\beq\label{PHQren}
\RPHQ (\VCOPCON) \, := \, \mathop{\lim}\limits_{\REGPAR \to 0} \, \RPHQ(\VCOPCON;\REGPAR) \,.
\eeq
The existence of such a formal diffeomorphism $\VCOPCON' = \VCHNG(\VCOPCON;\REGPAR)$ (\ref{CHNG}) for a given in advance series $\PHQ(\VCOPCON;\REGPAR)$ (\ref{PHQexp}) so that the limit (\ref{PHQren}) exists is far from being a trivial statement.
This phenomena is called {\it renormalizability} of $\PHQ$.
The physical interpretation of this procedure is that we pass by the change (\ref{CHNG}) to a new set of coupling constants called ``renormalized couplings'' so that the initial ``bare couplings'' become infinite (meaningless) for $\REGPAR \to 0$.

Still, the above renormalization procedure has a built in ambiguity. Namely, if we have one solution $\VCHNG(\VCOPCON;\REGPAR)$ (\ref{CHNG}) of this problem then any composition
$$
\VCHNG_1(\VCOPCON;\REGPAR) \, = \,  \VCHNG\bigl(\VFINREN(\VCOPCON);\REGPAR\bigr)
$$
with a formal diffeomorphism $\VFINREN(\VCOPCON)$
will also be a solution.
Thus, the group of formal diffeomorphisms of the couplings $\VCOPCON$ appears naturally as acting on the renormalization schemes.
This is the first notion of a renormalization group.
It is simply the group of formal diffeomorphism.

We see that the above concept of renormalization is rather general.
It leads also to the most primary concept of a renormalization group and so, it should be related to any other such notion.
More precisely, any other notion of a renormalization group should have a representation (a homomorphism) in the group of formal diffeomorphisms of the coupling constants.
In this case we speak about ``renormalization group action'', i.e., it is an action of the corresponding group by formal diffeomorphisms of the couplings.

\medskip

We pass now to a second notion of the renormalization group that is specific for QFT and it is finer than the above one.
In QFT there are additional technical features of the renormlization procedure.
Namely, each of the terms $\PHQ_{i_1,\dots,i_n}(\REGPAR)$ in series (\ref{PHQexp}) is additionally expanded in a finite sum labeled by a Feynman graph with $n$ vertices.
The renormalization adds to every diagram contribution a counter-term together with recursively determined counter-terms for subdiagrams.
Without going more into the details we will only mention that the ambiguity in the renormalization in QFT is described exactly by contraction maps introduced in Sect. \ref{S4.2n}.
So, we obtain now a finer notion of renormalization group that is formed by sequences of contraction maps.
One further shows that the composition in this group is exactly given by the rule following from the operadic structure on contraction maps.
The latter is shown in \cite[Sect. 2.6]{N1} in a more general context of renormalization than the graph-combinatorial one.

Thus, from this second perspective the renormalization group appears exactly as a group related to the contraction operad on Feynman diagrams.
Then, as explained above, there should be related a ``renormalization group action'', i.e., a homomorphism from this group to the group of formal diffeomorphisms of the couplings.
The existence and the derivation of this homomorphism follow also from the general renormalization theory and are not a part of the present work.
However, our result is that the resulting homomorphism corresponds to an operadic morphism via the functor established in Theorem \ref{Th1.1}.
Let us summarize all this:

\medskip

{\it
There is an operadic morphism, $\RMOR: \UROP_{\Adm,\TVER} \to \OEND_{\ \R^{\hspace{-1pt}\TVER}}$, from the contraction operad to the operad $\OEND_{\ \R^{\hspace{-1pt}\TVER}}$ over the vector space spanned by the set of type of vertices $\TVER$. The latter set indexes the set of coupling constants in the QFT model that is determined by the combinatorial data $(\Adm,\TVER)$. The induced map between the related groups
\beq\label{RRMOR}
\RRMOR \, : \, \GRP\bigl(\UROP_{\Adm,\TVER}\bigr) \to \GRP\bigl(\OEND_{\ \R^{\hspace{-1pt}\TVER}}\bigr)
\, \cong \
\FDIFF \bigl(\R^{\hspace{-1pt}\TVER}\bigr)
\,
\eeq
coincides with the renormalization group action determined from the renormalization theory.}

\medskip

In the subsequent subsections we will construct the morphism
$\RMOR: \UROP_{\Adm,\TVER} \to \OEND_{\ \R^{\hspace{-1pt}\TVER}}$.
We shall continue our considerations on a general ground field (ring) $\GRF$ but the above application uses the case $\GRF=\R$.

\subsection{Bosons and fermions}\label{Ss5.2xx}

We introduce a subdivision of the set of fields, i.e. the set $\FCL$ of flags' colors, into two disjoint subsets called bosons and fermions. 
According to this we assign $(\Z/2\Z)$--parities to the set of generators (\ref{NSGEN}) of the monoid $\MND(n)$.
For a bosonic $\phi$ the element $\phi(i)$ is even and for fermionic $\phi$, $\phi(i)$ is odd. 
The parity of the propagator $\CVR{\phi}{i}{\psi}{j}$ is the sum of the parities of the coupled fields $\phi$ and $\psi$. Usually bosons are coupled only to bosons and fermions - to fermions, so that the propagators are then always even. 
Finally, the interactions $L(i)$ are even as well.

Recall that we introduced in Sect. \ref{Se3} g a canonical isomorphism $\Dgm(n) \cong \MND(n)$ between the set $\Dgm(n)$ of all classes of isomorphic enumerated colored diagrams with $n$ vertices and the elements in the free monoid $\MND(n)$ generated by the set (\ref{NSGEN}).
Let us introduce the linear envelope of the monoid $\MND(n)$:
\beq\label{Mnd-1}
\Mnd(n) \, := \, \GRF^{(\MND(n))} \, \cong \, \GRF^{(\Dgm(n))} \,,
\eeq
which is thus an algebra.\footnote{% 
However, we remark that the algebra structure induced by the monoid structure of $\MND(n)$ is quite different from the algebra structure on the space of diagrams that is usually used in the Connes--Kreimer approach.}
In the more general case of presence of fermions we redefine the algebra structure on $\Mnd(n)$ (\ref{Mnd-1}) and set
\beq\label{Mnd-2}
\Mnd(n) \, := \, \text{the graded commutative algebra generated by the set (\ref{NSGEN}).}
\eeq

Note that in all the constructions up to now the division of the fields (i.e., the set $\FCL$) into bosons and fermions is inessential.

\subsection{The Wick generating operator of diagrams}\label{S4.6n} 

Let us assume first that we have a theory only with bosons so that the algebras $\Mnd(n)$ are commutative. 

Let us have $n$ vertices $\ver_1,\dots,\ver_n \in \TVER$ and consider them as one enumerated colored graph that is completely disconnected (i.e., it has no inner lines). The monomial in $\MND(n)$ corresponding to this diagram is thus $\ver_1(1) \cdots \ver_n(n) \equiv\ver_1 \otimes \cdots \otimes \ver_n$, where the number in bracket ``$(j)$'' indicates the number assigned to the corresponding vertex. Denote
\beqa\label{WIK0}
\Wck{\Adm}{n}(\ver_1,\dots,\ver_n) \, := \podr
\sum \text{all possible ways of connecting the vertices}  
\nnb \podr
\ver_1(1),\dots,\ver_n(n) \text{ into $\Adm$--admissible enumerated colored}
\nnb \podr
\text{graphs with no tadpoles}
\nnb = \podr \ver_1(1) \cdots \ver_n(n) + \cdots \,,
\eeqa
where $\Adm \subseteq \FCL^{\times 2}$ is a set of admissible connections as defined in Sect. \ref{S4.3n}.
This defines us a multilinear map
$$
\Wck{\Adm}{n} \, : \,
\bigl(\GRF^{\hspace{-1pt}\TVER}\bigr)^{\times n} \to \Mnd(n) \,.
$$

\begin{proposition}\label{Pro4.6}
{\rm (\cite{LN12})}
Under the isomorphism  $\Dgm(n) \cong \MND(n)$ $($Proposition \ref{Pro4.3}$)$ the following equation holds
\beqa\label{WIK}
\podr
\Wck{\Adm}{n}(\ver_1,\dots,\ver_n) 
\nnb
\podr = \,
\Biggl[
\mathop{\prod}\limits_{1 \, \leqslant \, i \, < \, j \, \leqslant \, n}
\exp 
\Biggl(
\mathop{\sum}\limits_{(\phi,\psi) \in \Adm} 
\CVR{\phi}{i}{\psi}{j} \frac{\di^2}{\di \phi(i) \di \psi (j)}
\Biggr) 
\Biggr] \,
\ver_1(1) \cdots \ver_n(n)
\,.
\eeqa
\end{proposition}

In the presence of fermions Eq. (\ref{WIK}) continues to generate the terms in the right hand side of Eq. (\ref{WIK0}) but with some signs that depend on the order of writing of the remaining generators of $\Mnd(n)$.
The derivatives $\frac{\di}{\di \phi(i)}$ are understood as left Grassman derivatives for odd $\phi(i)$.

\subsection{Construction of operadic morphism $\RMOR$ $:$ $\UROP_{\Adm,\TVER}$ $\to$ $\OEND_{\ \GRF^{\hspace{-1pt}\TVER}}$}

The operadic morphism $\RMOR$ $:$ $\UROP_{\Adm,\TVER}$ $\to$ $\OEND_{\ \GRF^{\hspace{-1pt}\TVER}}$  consists of a sequence of linear maps
\beq\label{RMORn}
\RMOR_n \, : \, \UROP_{\Adm,\TVER} (n) \to \OEND_{\ \GRF^{\hspace{-1pt}\TVER}} (n) \, \equiv \, \Hom \bigl(\bigl(\GRF^{\hspace{-1pt}\TVER}\bigr)^{\otimes n}, \GRF^{\hspace{-1pt}\TVER}\bigr)
\,.
\eeq
The ansatz for $\RMOR_n$ is
\beq\label{RMORn1}
\RMOR_n (Q) (\ver_1 \otimes \cdots \otimes \ver_n)
\, = \,
\widehat{Q} \Bigl(\Wck{\Adm}{n}(\ver_1,\dots,\ver_n) \Bigr)
\, \in \, \GRF^{\hspace{-1pt}\TVER}
\,,
\eeq
where $Q \in \UROP_{\Adm,\TVER} (n) \subseteq \UROP (n)$ is generally given by Eq. (\ref{Frm}) and $\widehat{Q}$ is then set to be
\beq\label{Frm2}
\widehat{Q} (\GRPH) \, = \, 
\mathop{\sum}\limits_{L \, \in \, \VCL}
q(\GRPH,L) \, \bigl[\GRPH\bigl/(\VRT(\GRPH) ,L)\bigr] 
\,, \qquad
\widehat{Q} : \GRF^{(\Dgm(n))} \to \GRF^{\Dgm(1)}
\,,
\eeq
i.e., $\widehat{Q} (\GRPH)$ contracts the diagram $\GRPH$ to a sum of single vertices according to the color prescription of $Q : \GRF^{(\Dgm(n))} \to \GRF^{(\VCL)}$.

Let us explain by words the meaning of Eq. (\ref{RMORn}).
The value of $\RMOR_n (Q) (\ver_1 \otimes \cdots \otimes \ver_n)$ is a sum of single vertices obtained by making first a sum over all possible ways of connecting the vertices $\ver_1(1),\dots,\ver_n(n)$ into enumerated diagrams; then we contract each of the terms in the latter sum to a sum of single vertices via $Q$. Shortly speaking, $\RMOR_n (Q) (\ver_1 \otimes \cdots \otimes \ver_n)$ is the $Q$--contraction of all possible connections of $\ver_1,\dots,\ver_n$ into diagrams.

\begin{proposition}\label{pr-x}
{\rm (\cite{LN12})}
Equation $(\ref{Frm2})$ determines an operadic morphism.
\end{proposition}

\section{Outlook}\label{Se-Y}

We make here a connection with the Connes--Kreimer Hopf algebra of ``formal diffeographisms'' (\cite{CK}), which in details will appear in a forthcoming work.

The first step towards the comparison with the Connes--Kreimer approach is to study the dual (commutative) Hopf algebra to the Lie algebra associated with a symmetric operad.
In fact, it can be associated directly to a symmetric \emph{co-operad}.
When this construction is applied to the contraction operads on diagrams we obtain a Hopf algebra that is very close to the Connes--Kreimer Hopf algebra.
However, there is an important difference. On a technical level, in our approach a subdiagram is always contracted to a vertex, while in the Connes--Kreimer theory some subdiagrams that have two external lines can be contracted also to an \text{edge} with no intermediate vertex.

The origin for this difference comes from physics.
The Connes--Kreimer Hopf algebra incorporates an additional step in the renormalization called a ``field renormalization''.
Let us briefly explain this. 
Our set of vertices $\TVER$ corresponds to all the monomials in the Lagrangian of a given QFT model. Some of these vertices of valence two correspond to quadratic terms in the Lagrangian, which are called ``kinetic terms'' since they basically determine the propagators.
For this reason in physics there are no physical parameters related to these terms: we always normalize them with some standard normalization coefficients like
$$
\frac{1}{2} (\di \phi) \cdot (\di \phi)
\,,\qquad
\overline{\psi}(\gamma \cdot \di)\psi
\,,
$$
for a scalar and a spinor field, respectively ($\gamma \cdot \di$ being the Dirac operator).
On the other hand, as a result of the renormalization the coefficients in front of these kinetic terms are changed (renormalized).
Then we absorb this change by a redefinition of the field strengths.
For instance, in the above examples we pass to new fields $\phi' = Z_{\phi} \phi$, $\psi' = Z_{\psi} \psi$ and $\overline{\psi}' = \overline{Z}_{\psi} \overline{\psi}$ so that the kinetic terms are changed by $Z_{\phi}^2$ and $\overline{Z}_{\psi}Z_{\psi}$, respectively, in such a way that compensate the renormalization change.

\begin{acknowledgement}
We thank Dorothea Bahns, Kurusch Ebrahimi-Fard, Alessandra Frabetti, Klaus Fredenhagen and Raymond Stora for fruitful discussions.
The work was partially supported by the French-Bulgarian Project Rila under the contract Egide-Rila N112.
N.N. thanks the Courant Research Center "Higher order structures in mathematics" (G\"ottingen) and the II. Institute for Theoretical Physics at the University of Hamburg for support and hospitality. 
\end{acknowledgement}

\setcounter{section}{0}
\renewcommand{\thesection}{\Alph{section}}

\end{document}